\documentclass{llncs}

\usepackage[lined,boxed,commentsnumbered]{algorithm2e}

\newcommand{\remove}[1]{}

\newcommand{\ignore}[1]{}

\begin{document}

\bibliographystyle{plain}


\title{\textbf{Snapshot for Time: The One-Shot Case}}
\author{Eli Gafni}
\institute{
Computer Science Department, University of California,
Los Angeles, CA 90024, USA
}
\maketitle

\begin{abstract}

We show that for one-shot problems - problems where a processor executes a single operation-execution - timing constraints can be captured by conditions on the relation between original outputs and supplementary snapshots. In addition to the dictionary definition of the word snapshot, in distributed computing snapshots also stand for a task that imposes relation among sets which are output of processors. Hence, constrains relating the timing between operation-executions of processors can be captured by the sets relation representing a task.

This allows to bring to bear techniques developed for tasks, to one-shot objects. In particular, for the one-shot case
the question of linearizability is moot. Nevertheless, current proof techniques of object implementation require the prover to provide linearization-points even in the one shot case. Transforming the object into a task relieves the prover of an implementation from the burden of finding the ``linearization-points,'' since if the task is solvable, linearization points are guaranteed to exist. We exhibit this advantage with a new algorithm to implement MWMR register in a SWMR system.

\end{abstract}






\newpage

\section{Introduction}

Separation of concerns is the key to tackling a complex job. It applies as well as an essential strategy of attacking an involved research problem. Case in point is the question of the understanding of linearizable objects \cite{HW,LA}. This question has many moving elements:
\begin{enumerate}
\item Solvability: Given a model and an object specification is it implementable, even non-blocking?
\item If implementable non-blocking, is it implementable wait-free?
\item If implementable wait-free is the solution wait-free linearizable in the sense that every operation-execution ($opex$) can be thought of as happened in an instant inside the interval of the $opex$?
\end{enumerate}

While many implementations of linearizable objects have appeared \cite{SWAP,STACK,Common2}, our understanding of how to go about them is lacking. It is attested to by two long-standing apparently unresolved questions: Does there exist a wait-free implementation of a queue from 2-processors consensus \cite{Common2}, and is the so-called Herlihy Hierarchy robust \cite{ROBUST} with respect to deterministic objects, the only partially convincing arguments of \cite{SENSE} and \cite{Neiger}, not withstanding.

We propose to apply separation of concerns to tackle ``objects,'' and in general, timing constraints, by concentrating first on the one-shot case. In the one-shot case, we have finite number of processors and each executes one a priori given ``command'' on a shared object. We will not clutter this paper with $k$-shots generalization, but transforming one-shot to $k$-shot is a standard faire. Even in the one-shot case of an object  the problem is challenging. The implementation of one-shot SWAP, defined later, is a challenging exercise to anybody who has not seen the solution beforehand.

If we do not know one-shot object implementation to be necessarily easy, then, are one shot-objects as hard in general as tasks? The question of solvability of tasks is in general undecidable \cite{GK,HR}.
Are one-shot objects undecidable? As we argue later, one-shot objects are a strict subclass of tasks. Thus, they might be decidable. In fact we conjecture they are! But nevertheless just the fact that such a straight forward question is still on the table shows we haven't started to scratch the surface of understanding objects, all the aforementioned  ingenious implementations in \cite{SWAP,STACK,Common2} not withstanding.

In this short paper we do a promising step in the hope of sheding further light on objects. We show we can transform the question of solvability of one-shot object, into the question of solvability of a task. For tasks we have developed a deep understanding by bringing topology to bear. Hence the hope that this reformulation is a step in the right direction of gaining an understanding.

Tasks are simple mathematical objects formulated as relations. They do not have built in facility to capture the ``happened-before'' relationship introduced by Lamport \cite{opex}, that underly all object specifications. This paper shows the round-about-way of adjoining an object specification with the task of snapshots, to let the object be specified as a task.

The idea is pretty simple. In an implementation of one-shot object a processor executes a single $opex$. The $opex$ starts when the processor takes the first step in the algorithm, and ends when it outputs. In the models we consider, which are restrictions of read-write shared-memory protocol, the $opex$ would start with a processor registering in shared-memory. Similarly to the specification of task solvability, this registration makes the processor participating, and its participation can be detected by read operation by other processor.
At the end of an operation execution, we prolong it in a wait-free manner to have a processor do a collect \cite{Collect} or for more structure and being succinct, an atomic snapshot \cite{ATOMIC}. If processor $p_j$ started its $opex$ after $p_i$ finished its $opex$, then, as a result as prolonging the operation-execution by a snapshot, $p_j$ would not appear in $p_i$'s snapshot. We show that this is sufficient to capture all happen-before timing constraints. 

When we jet pose the question of how to specify a multi-writer multi-reader (MWMR) register object as a task, we obtain the gist of the implementation in \cite{MWMR}. The algorithm is almost ``forced'' by the specification of MWMR as a task. We contend that the same hold for object of 2-processors consensus power, when one does away with the use of 2-processors consensus to achieve gets-and-set, but rather when one uses the cumulative set consensus power of the 2-processors consensus, to narrow the effective number of active processors using set-consensus state-machine and then use wait-free like methods. But that's beyond the scope of this paper hence we settle with the single non-trivial consensus power 1 object we know of.

The one shot tasks and objects have their non-terminating counterparts. In that case we have the question of wait-freedom. Is the solution or the implementation is such that no starvation occurs even though the system progresses. In the case of objects we have a question beyond wait-free implementation and that is the question of linearization. An evolving solution dictates a partial order between $opex$'s. A possible total order of the partial order determines the viability of the partial solution. The question is whether one can commit growing and growing single total order of a prefix of the solution as the solution produces more and more outputs. This is the question called ``linearizability'' \cite{HW, LA}. This question is moot in the one-shot case, nevertheless, the techniques of implementations in \cite{HW,LA} force us to deal with ``linearization-points'' even in the one-shot case. As we show, the reformulation of one-shot object as a task relieves us off worrying about ``linearization-points.'' The solvability of resultant task proved correct in any technique implies that the linearization points exist. We exhibit this point with a new implementation and proof of correctness of a one shot MWMR implementation.

This idea of use of a snapshot is probably hidden also in all the implementations in \cite{SWAP,STACK,Common2,MWMR}, but was never isolated and pinpointed as a module - a module that transforms an object specification into a task. This is not unlike the Atomic-Snapshot formulated in \cite{ATOMIC,composite}. Numerous examples exist (\cite{Ambiguous} as classic example we know of) of researchers ``doing snapshots implicitly'' inside their algorithms prior to \cite{ATOMIC,composite} without taking stock of what they have done, and isolating it as a module.

\section{Preliminaries}
\subsection{Tasks}
A task on $n+1$ processors $P=\{p_0,\ldots,p_n\}$, is a mathematical triple $(I,O,\Delta)$ independent of any model.
$I$ and $O$, called input and output tuples, respectively, are sets of pairs, each pair is of the for $(processor-id, ~value)$
where $value$ comes from some set of values. The sets in $I$ and $O$ range over the whole $ 2^{n+1}$ subsets of the processors. All the pairs in a set contain distinct $processor-id$.

$\Delta$ is a binary relation between $I$ and $O$ with the only constraint that it associates sets from $I$ and $O$ that range over the same set of $processor-id$.

To make things concrete we describe the adaptive renaming task \cite{RENAMING}. The set $I$ is all the distinct $2^{n+1}$ sets each
entry of whom is of the form $(procesor-id, \bot)$. This means that the input to a processor is just its $processor-id$. Such a task is called $inputless$-task, as the $processor-id$ is always understood as part of the input.

For the input tuple $\{ p_i \in E ~|~(p_i,\bot)\}$ for some $E \subseteq P$ and $|E|=k$, $\Delta$ associates any set
of the form $\{p_i \in E~|~(p_i,integer_i)\}$ where $1 \leq integer_i \leq 2k-1$ and $\forall i,j,~ integer_i \not = integer_j$.

In the input-less SWAP task, one processor returns $\bot$ and the rest return each a processor-id from $E$. When we draw a directed graph whose nodes are processor-ids and $\bot$ and draw a directed edge from a processor-id to the value it returns we obtain a simple directed path.

\subsection{Models}

In this paper we restrict the notion of a model to be any subset of infinite runs of the SWMR asynchronous wait-free model \cite{opex}.
It was established in \cite{ramifications} that for any set-consensus object there exits a set of runs that has precisely the power of the object. Thus, a reader who likes wait-free model with objects may still retain her/his mode of thinking.

We chose this restriction since for these models the notion of ``participation'' of a processor is well defined. W.l.o.g. the first operation of a processor in a run is to write its input to its dedicated cell. One it does this it is called \emph{participating} in the run. A processor that reads and writes infinitely many times in a run in the model is called
\emph{live}.

Examples: The $t$-resilient model is the models of all runs in which at least $(n+1)-t$ processors read and write infinitely often. The model that has the exact power of 2-processors consensus is the model of all runs such that when projected on any live $p_i,p_j$ the interleaving of the two processors breaks symmetry infinitely often.

\subsection{Solving a Task in a Model}

A task $T$ is solvable in $M$ if there exits an algorithm $A$ in $M$, and a partial map from local states of processors in runs of $A$ to an output such that:
\begin{enumerate}
\item For any live $p_i$ with input $i_i$ in a run $r \in M$ there is a local state generated in $A$ that maps to an output $o_i$.
\item For the participating set $P_r$ of $r$, corresponding to the input-tuple $\{ p_j \in P_r~|~ (p_j,i_j\}$ there exists an output-tuple which agree on the outputs of the live processors.
\end{enumerate}

\subsection{One-shot Objects}

A one shot-object over the set of processors $P=\{p_0,\ldots,p_n\}$ is a state-machine over a set of $commands$ $C$ that specifies a state-transition function and response for any given $command$ at a given state.  When the function is non-deterministic the object is non-deterministic, otherwise, it is deterministic.

An example of an object is a MWMR register. It accepts commands of the form $read$ or $write(v)$, where $v$ is a value. Its state-transition for a $read$ is to remain at the same state. Its response it the value it inductively holds. Its state transition to $write(v)$ is to change its state to holding $v$, and respond with $ok$.

A queue accepts commands $enqueue(x)$, where $x$ is an abstract $item$, and the command $dequeue$, inductively the queue holds a list of abstract value. It state-transition for $enqueue(x)$ is to append $x$ to the tail of the list, and its response is $ok$. Its state-transition as a result of $dequeue$ is to remove an item from the head-list if the list is non-empty, else remain in the same state. Its response is the item if it exists or $\bot$ it it does not.

The object SWAP is like the task SWAP only that if $op_i$ occurred before $op_j$ the edge $op_i$ is closer to $\bot$ than the edge corresponding to $op_j$.

\subsection{Implementing an Object in a model $M$}

To implement a one-shot object in $M$ is to have an algorithm for all processors initialized with commands and the object initialized to some initial state, and a partial map from views in $A$ to outputs, such that:
\begin{enumerate}
\item For any live $p_i$ with input $i_i$ in a run $r \in M$ there is a local state generated in $A$ that maps to an output $o_i$.
\item An $op$ of a processor in $r \in M$ starts when the processor starts participating until it outputs. The $op$ continues ad infinity if no view maps to output. Then there exist a linearization of the happened before relationship, such that the linear behavior of the object agrees with the outputs of the live processors.
\end{enumerate}
Thus, the only difference in the specification of solving a task or an object is that the solution of an object cares about the relative order of $op$'s.

\subsection{Tasks cannot be formulated as Objects}

This paper will proceed to establish that one-shot object is a special case of a task. Can any task be posed as an object possibly non-deterministic? The answer is negative. Consider the task of adaptive renaming. In solving adaptive renaming there may be no processor that outputs the integer 1. That is because a task may look at all the invocation ``until now'' and respond accordingly. An object cannot ``see'' the future. If it will not give out the value 1 in all runs, it will mean it predicted the future that the run will not be a solo run. It is thus in our opinion meaningless to be called an ``object,'' in such a case.

\subsection{The Atomic-Snapshots (AS) Task}

If the AS task is executed by operation execution if we order the snapshots, what is the condition on the order of snapshots that makes this order a total order of the partial order of the operation execution? The only real timing constraint on the snapshot task is that if $p_i$ finished its operation-execution before $p_j$ started it operation execution then the output set of $p_i$, $S_i$ is such that $p_j \not \in S_i$. This is equivalent to the following:

For a sequence of processors $p_i$ with their output sets $S_i$ we say that the sequence satisfies the \emph{well-ordering} if any position $k$ in the sequence, the intersection of $S_j$ $j \geq k$ is greater equal than union $\{p_m\}$ $m<k$.
\begin{proposition}
For every sequence of processors and their snapshots that satisfies well-ordering there exist a $opex$'s of the wait-free Atomic-Snapshots in \cite{ATOMIC} whose happened-before relation satisfies the linear order of the sequence. Conversely, for a sequence that does not satisfy well-ordering, no algorithm can have a partial order on $opex$'s that is in compliance with the total order imposed by the sequence.
\end{proposition}

The proof follows from the definition of a snapshot. The point of the proposition is to say that if someone poses the snapshot problem through timing constraints then what the snapshot task formulation and implementation in \cite{ATOMIC} did ``without knowing,'' is take these timing constrains and pose them as a task. 

\section{Ordered-Tasks}
An ordered task is a relation between sequences of processors to sequences of outputs:
For a sequence of length $l$ of distinct processors it a assigns a set of sequences of length $l$ of outputs.
For instance we exhibit what we call the ordered-task of adaptive renaming: For a sequence of $l$ of distinct processors it assigns all possible sequences of $l$ distinct integers from $\{1,\ldots,2l-1\}$ such that the sequence is increasing.

The adaptive renaming ordered-task wants to capture the additional timing constraint on the task of adaptive-renaming \cite{RENAMING} that says that a processor $p_i$ whose $opex_i$ is after $opex_j$ by processor $p_j$, should output an integer larger than the one output by $p_j$.

An ordered task is solvable in $M$ if there exist an algorithm in $M$ whose operation execution can be made into a total order of the processors such that their outputs are increasing, i.e. the partial order of $opex$'s can be made into a sequence of processors executing the $opex$'s (input sequence) and their outputs ordered into a sequence is an output sequence.

Does the order-task of adaptive-renaming solvable wait-free? Indeed, the wait-free adaptive renaming in \cite{BG93} as stated solves the ordered-task adaptive renaming. We state without elaborating on the proof that follows from \cite{BG93} and \cite{Raynal}, that the minimal space adaptive-renaming solvable task given $k$-test-and-set, is the same space that solves the ordered-task version.

\subsection{Reducing Ordered-Tasks to Tasks}

\begin{proposition}
The ordered task of adaptive renaming is solvable in $M$ iff the adaptive renaming task with each processor outputting a snapshot too, is solvable such
\begin{enumerate}
\item The integers output are in the range prescribed by the adaptive renaming, and
\item If the snapshots of the output pairs $(p_i,(integer_i,snapshot_i))$ are ordered by $integer_i$, then the snapshots are well-ordered.
\end{enumerate}
\end{proposition}

Interestingly, we remark in passing, the question of linearizability can now apply to ordered-tasks in general.
Linearizability was raised in the context of objects, where thinking about transition happening in an instance make sense. But once we translate it to task we see it is not really about the ``instance.'' It is about committing to a total-order linearization of growing and growing prefix of the partial order solution. A generalization that should occupy a paper by itself.

Thus, if the non-terminating ordered task of adaptive renaming is wait-free solvable, then it is linearizably wait-free solvable, since in this ordered-task the input sequence is a function of the output values. 

\subsection{Objects as Ordered-Tasks}

It is easy to see that a one-shot object is an ordered task, where the specification of the viability of an output tuple-to an input-tuple is whether the application of the sequence of commands as specified by the input will produce the sequence of responses specified by the state machine. Thus it is also a map from sequences to sequences, hence an ordered task.

Thus we have the pinnacle Proposition of this paper:
\begin{proposition}
A one-shot object O is implementable in $M$, if and only if the following task $TO$ is solvable in $M$:
\begin{enumerate}
\item The set $I$ for $TO$ is the set of pairs $(p_i,command_i)$ from $O$. Thus $TO$ is actually input-less, since there is single $command$ associated with each processor $p_i$.
\item The output sets $O$ of $TO$ hold entrees of the form $(p_i,(value_i,snapshot_i))$.
\item $\Delta$: An output tuple on a set $E$ is valid output iff there exists an ordering of the $p_i \in E$, such that the induced ordered $value$'s is the sequence of outputs $O$, and this sequence is well-ordered with respect to the snapshots.
\end{enumerate}
\end{proposition} 
\begin{proof}
The only non-trivial direction is $(\Leftarrow)$: Given a solution to the task we have to show linearization points. We do it inductively. Let $p_i$ be the last in the order. Its snapshot is $S_i$. Place the linearization at the end of $opex_i$.
All we need to show to continue the induction is that no other $opex$ started after the end of $opex_i$. But since by well ordering $S_i$ contains all processors in $E$, a processor in $E$ that started its $opex$ after $opex_i$ will not appear in $snapshot_i$, contradiction.
\end{proof}

\subsubsection{Example of the Implementation of NWMR register Based on the Snapshot Formulation}

The idea based on the snapshot formulation is for a processor to post a snapshot, and then take another snapshot to finish its $op$. Thus, with an $op$ we have two snapshots. One, which we will call \emph{early-snapshot} is visible to all. The latter, \emph{late-snapshot} is part of the individual output of a processor in a task, and it is the one the correctness is determined by. When the early-snapshot is posted it establishes a ``conservative'' partial order. The final partial order is a sub-order of this partial order. Therefore, working according to the conservative partial order when it is visible, will guarantee correctness. The approach is best manifested by the following two examples:

Implementation of one-shot MWMR register by SWMR wait-free:\\
$Write(v)$ by $p_i$: Post $id$ in shared-memory. Snapshot. Write$(v,early-snapshot)$. Snapshot. Output: $(ok,late-snapshot)$.\\
$Read$ by $p_j$: Post $id$ in shared-memory. Snapshot. Total-order (by some consistent tie breaking rule) the $op's$ whose early-snapshot is posted. Return $v$ from the latest $op$. Snapshot. Output: $(v, late-snapshot)$.\\

\begin{algorithm}[t]

{\small
Shared Array $Id[ 0 \ldots n]$ initialized to $\emptyset$\;
Shared Array  $Value[ 0 \ldots n]$  initialized to $\bot$\;
\
\
\\
%
%
%
%
{\bf Write(v):}{\\
{$Id[i]:=id$\;}
{$early-snapshot:= \cup Id$\;}
{$Value[i]:=(v,early-snapshot)$\;}
{$late-snapshot:=\cup Id$\;}
{{\bf Return}$(ok,late-snapshot)$\;}
}
\
\
\\

{\bf Read:}{\\
{$Id[i]:=Id$\;}
{$snapshot:=Value$\;}
{local:$~~j:=\{p | Value[p] \not = \bot, \forall q:Value[q] \not = \bot (|Value[q].early-snapshot| < |Value[p].early-snapshot|  \vee ( |Value[q].early-snapshot| = |Value[q].early-snapshot| \wedge q <p))$\;}
{$late-snapshot:=\cup Id$\;}
{{\bf Return}$(Value[j].v~,late-snapshot)$\;}
}



}

\caption{ MWMR One-Shot Implementation}
\label{alg:lin}
\end{algorithm}

\begin{proposition}
Algorithm 1 solves the one-shot MWMR timed-task.
\end{proposition}

\begin{proof}
We prove by backward induction that we can totally order the output values in compliance with the well ordering of the returned snapshots.

W.l.o.g a Writer processor $q$ writes $q$
Let processor $p$ have Write $op$  that posted the largest $early-snapshot$ and has the highest id among the Writers that posted the largest early-snapshot. Let $rest$ be set of all processors but those Reader that returned $p$, and $p$ itself.
The backward ordering is linear ordering of $\{Readers-returning-p\},p,rest$.

We first show that $rest \subseteq p.late-snapshot$: By the algorithm, a Reader $p_j \in rest$ did not $snapshot$ the $early-snapshot$ of $p$.  Since $p_j.snapshot$ is after $p_j$ posted its id, and $p.late-snapshot$ is after it $wrote$ its $early-snapshot$ we conclude that $p_j \in p.late-snapshot$.

To see that same holds for a Writer $p_j \in rest$ notice that by definition $p_j$ has an $p_j.early-snapshot \leq p.early-snapshot$.
By the property of snapshot and the fact that $early-snapshot \subseteq late-snapshot$ we get $p_j \in p.late-snapshot$.

Since $q \in \{Readers-returning-p\}$ has $snapshot$ $p.early-snapshot$ its $snapshot$ occurred after $p$ $wrote$ $p.early-snapshot$ and therefore we get $p_j \in q.late-snapshot$.

Now order $\{Readers-returning-p\}$ by their $late-snapshots$'s.

Continue inductively on $rest$.

\end{proof}

\section{Conclusions}

We reduce the apparent ``entropy'' of distributed computing with its apparently unrelated diverse models, problems, communication mechanisms, by showing that the creature called ``object,'' at least when considered in the one-shot version, is a task!
What about wait-freedom in the infinite-case etc. etc.? \\
Rome was not built in one-day. You build it one stone at a time. Separation of concerns: First understand the finite case.

We even show that beyond the aesthetic satisfaction of reducing entropy, viewing objects as tasks for the one shot case makes implementations and their proof clearer. When we compare the MWMR implementation of Algorithm1 with the implementation and proof in \cite{MWMR}, we do see a value.


\end{document}